\documentclass[11pt,a4paper]{article}

\usepackage{amsmath,amsfonts,amssymb,amsthm}

\usepackage{backref}

\usepackage{graphicx,color}
\usepackage{boxedminipage}
\newtheorem{theorem}{Theorem}

\newtheorem{corollary}{Corollary}
\newtheorem{lemma}{Lemma}

\DeclareMathOperator{\operatorClassNP}{NP}
\newcommand{\classNP}{\ensuremath{\operatorClassNP}}
\DeclareMathOperator{\operatorClassCoNP}{coNP}
\newcommand{\classCoNP}{\ensuremath{\operatorClassCoNP}}
\DeclareMathOperator{\operatorClassFPT}{FPT}
\newcommand{\classFPT}{\ensuremath{\operatorClassFPT}}
\DeclareMathOperator{\operatorClassW}{W}
\newcommand{\classW}[1]{\ensuremath{\operatorClassW[#1]}}

\begin{document}

\title{Editing to a Graph of Given Degrees\footnote{The research leading to these results has received funding from the European Research Council under the European Union's Seventh Framework Programme (FP/2007-2013)/ERC Grant Agreement n. 267959. Preliminary version of the paper appeared in the proceeding of IPEC 2014.}}

\author{Petr A. Golovach\thanks{Department of Informatics, University of Bergen, PB 7803, 5020 Bergen, Norway. E-mail: {\tt{petr.golovach@ii.uib.no}}}}

\date{}

\maketitle

\begin{abstract}
 We consider the \textsc{Editing to a Graph of Given Degrees} problem that asks for a graph $G$, non-negative integers $d,k$ and a function $\delta\colon V(G)\rightarrow\{1,\ldots,d\}$, whether it is  possible to obtain a graph $G'$ from $G$ such
 that  the degree of $v$ is $\delta(v)$ for any vertex $v$ by at most $k$ vertex or edge deletions or edge additions.
We construct an \classFPT-algorithm for  \textsc{Editing to a  Graph of Given Degrees}  parameterized by $d+k$. We complement this result by showing that the problem has no polynomial kernel unless  $\classNP\subseteq\classCoNP/\text{\rm poly}$.
\end{abstract}

\section{Introduction}
The aim of graph editing or modification problems is to change a given graph  
by applying a bounded number of specified operations in order to satisfy a certain
property. 
Many basic problems like {\sc Clique}, {\sc Independent Set} or {\sc Feedback (Edge or Vertex) Set}  can be seen as graph editing problems. 
It is common to allow combinations of vertex deletions, edge deletions and edge additions,  but other operations, like edge contractions, are considered as well.

The systematic study of the vertex deletion problems was initiated by Lewis and Yannakakis~\cite{LewisY80}. They considered hereditary non-trivial properties. 
A property is hereditary if it holds for any induced subgraph of a graph that satisfy the property, and a property is non-trivial if it is true for infinitely many graphs and false for infinitely many graphs.  Lewis and Yannakakis~\cite{LewisY80} proved that for any non-trivial hereditary property,  the corresponding vertex deletion problem is \classNP-hard, and for trivial properties the problem can be solved in polynomial time.
The edge deletion problems were considered by Yannakakis~\cite{Yannakakis78},   Alon, Shapira and Sudakov~\cite{AlonSS05}. The case when edge additions and deletions are allowed and the property is the inclusion in some hereditary graph class was considered by Natanzon, Shamir and Sharan~\cite{NatanzonSS01} and Burzyn, Bonomo and Dur{\'a}n~\cite{BurzynBD06}.

As typically graph editing problems are \classNP-hard, it is natural to use the parameterized complexity framework to analyze them. Cai~\cite{Cai96} proved that for any property defined by a finite set of forbidden induced subgraphs, the editing problem is \classFPT\ when parameterized by the bound on the number of vertex deletions, edge deletions and edge additions. Further results for hereditary properties were obtained by Khot and Raman~\cite{KhotR02}.

As it could be seen from the aforementioned results, the editing problems are well investigated for hereditary properties. For properties of other types,  a great deal less is known, and the graph editing problems where the aim is to obtain a graph that satisfies degree constraints belong to the class of graph editing problems for non-hereditary properties. Investigation of the parameterized complexity of such problems were initiated
by Moser and Thilikos in~\cite{MoserT09}, Cai and Yang~\cite{CaiY11} and Mathieson and Szeider~\cite{MathiesonS12} (see also~\cite{CyganMPPS11,Golovach14a} for related  results).

In particular,
Mathieson and Szeider~\cite{MathiesonS12} considered different variants of  the following problem:
\begin{center}
\begin{boxedminipage}{.99\textwidth}
\textsc{Editing to a Graph of Given Degrees}\\
\begin{tabular}{ r l }
\textit{~~~~Instance:} & A graph $G$, non-negative integers $d,k$ and a function\\
                                  & $\delta\colon V(G)\rightarrow\{1,\ldots,d\}$.\\
\textit{Parameter 1:} & $d$.\\                                  
\textit{Parameter 2:} & $k$.\\
\textit{Question:} & Is it possible to obtain a graph $G'$ from $G$ such
 that\\ &  $d_{G'}(v)=\delta(v)$ for each $v\in V(G')$ by at most $k$\\ 
 & operations from the set $S$?\\
\end{tabular}
\end{boxedminipage}
\end{center}
They classified the parameterized complexity of the problem for $$S\subseteq\{\text{vertex deletion},\text{edge deletion},\text{edge addition}\}.$$ 

They showed that \textsc{Editing to a Graph of Given Degrees} is \classW{1}-hard when parameterized by $k$ and the unparameterized version is \classNP-complete if  vertex deletion is in $S$.
If $S\subseteq\{\text{edge deletion},\text{edge addition}\}$, then the problem can be solved in polynomial time.   
For  $\{\text{vertex deletion}\}\subseteq S\subseteq \{\text{vertex deletion},\text{edge deletion},\text{edge addition}\}$, they proved that 
\textsc{Editing to a Graph of Given Degrees} is \emph{Fixed Parameter Tractable} (\classFPT) when parameterized by $d+k$. Moreover, the FPT result holds for a more general version of the problem where vertices and edges have costs and the degree constraints are relaxed: for each $v\in V(G')$, $d_{G'}(v)$ should be in a given set $\delta(v)\subseteq \{1,\ldots,d\}$. 
The proof given by Mathieson and Szeider~\cite{MathiesonS12}  uses a logic-based approach that does not provide practically feasible algorithms.
They used the observation that  \textsc{Editing to a Graph of Given Degrees} can be reduced to the instances with graphs whose degrees are bounded by a function of $k$ and $d$.  
By a result of Seese~\cite{Seese96}, the problem of deciding any property that can be expressed in  first-order logic
is \classFPT\ for graphs of bounded degree when parameterized by the length of the sentence defining the property. In particular, to obtain their \classFPT-result, Mathieson and Szeider constructed a non-trivial first-order logic formula that expresses the property that a graph with vertices of given degrees can be obtained by at most $k$ editing operations. 
%They left open the question whether a direct \classFPT-algorithm for the problem with a ``reasonable'' dependence on $k+d$ could be obtained. 
For the case $S\subseteq\{\text{vertex deletion},\text{edge deletion}\}$, they improved the aforementioned result by showing that 
\textsc{Editing to a Graph of Given Degrees} has a polynomial kernel when parameterized by $d+k$.  Some further results were recently obtained by 
Froese, Nichterlein and Niedermeier~\cite{FroeseNN14}.

In Section~\ref{sec:FPT} we construct an \classFPT-algorithm for \textsc{Editing to a Graph of Given Degrees} parameterized by $k+d$ for the case when $S$ includes
vertex deletion and edge addition that runs in time $2^{O(kd^2+k\log k)}\cdot poly(n)$ for $n$-vertex graphs, i.e., we give the first constructive algorithm for the problem.  
Our algorithm is based on the random separation techniques introduced by Cai, Chan and Chan~\cite{CaiCC06}. 
We complement this result by showing in Section~\ref{sec:no-ker} that  \textsc{Editing to a Graph of Given Degrees}  parameterized by $k+d$ has no polynomial kernel unless  $\classNP\subseteq\classCoNP/\text{\rm poly}$ if $\{\text{vertex deletion},\text{edge addition}\}\subseteq S$. 
This resolves an open problem by Mathieson and Szeider~\cite{MathiesonS12}. The proof uses the cross-composition framework introduced by Bodlaender, Jansen and Kratsch~\cite{BodlaenderJK14}.

\section{Basic definitions and preliminaries}\label{sec:defs}

\noindent
{\bf Graphs.}
We consider only finite undirected graphs without loops or multiple
edges. The vertex set of a graph $G$ is denoted by $V(G)$ and  
the edge set  is denoted by $E(G)$.

For a set of vertices $U\subseteq V(G)$,
$G[U]$ denotes the subgraph of $G$ induced by $U$, and by $G-U$ we denote the graph obtained form $G$ by the removal of all the vertices of $U$, i.e., the subgraph of $G$ induced by $V(G)\setminus U$. If $U=\{u\}$, we write $G-u$ instead of $G-\{u\}$. 
Respectively, for a set of edges $L\subseteq E(G)$, $G[L]$ is a subgraph of $G$ induced by $L$, i.e, the vertex set of $G[L]$ is the set of vetices of $G$ incident to the edges of $L$ and $L$ is the set of edges of $G[L]$.  
For a non-empty set $U$, $\binom{U}{2}$ is the set of unordered pairs of elements of $U$.
For a set of edges $L$,
by $G-L$ we denote the graph obtained from $G$ by the removal of all the edges of $L$.
Respectively, for $L\subseteq \binom{V(G)}{2}$, $G+L$ is the graph obtained from $G$ by the addition of the edges that are elements of $L$.
If $L=\{a\}$, then for simplicity, we write $G-a$ or $G+a$.

For a vertex $v$, we denote by $N_G(v)$ its
\emph{(open) neighborhood}, that is, the set of vertices which are adjacent to $v$, and for a set $U\subseteq V(G)$, $N_G(U)=(\cup_{v\in U}N_G(v))\setminus U$.
The \emph{closed neighborhood} $N_G[v]=N_G(v)\cup \{v\}$, and for a positive integer $r$, $N_G^r[v]$ is the set of vertices at distance at most $r$ from $v$. 
For a set $U\subseteq V(G)$ and a positive integer $r$, $N_G^r[U]=\cup_{v\in U}N_G^r[v]$.
The \emph{degree} of a vertex $v$ is denoted by $d_G(v)=|N_G(v)|$.

A  \emph{walk} in $G$ is a sequence $P=v_0,e_1,v_1,e_2,\ldots,e_s,v_s$ of vertices and edges of $G$ such that $v_0,\ldots,v_s\in V(G)$, $e_1,\ldots,e_s\in E(G)$,
and for $i\in\{1,\ldots,s\}$, $e_i=v_{i-1}v_i$; $v_0,v_s$ are the \emph{end-vertices} of the walk, and $v_1,\ldots,v_{s-1}$ are the \emph{internal} vertices. 
A walk is \emph{closed} if $v_0=v_s$. 
Sometimes we write $P=v_0,\ldots,v_s$ to denote a walk $P=v_0,e_1,\ldots,e_s,v_s$ omitting edges.
A walk is a \emph{trail} if $e_a,\ldots,e_s$ are pairwise distinct, and a trail is a \emph{path} if $v_0,\ldots,v_s$ are pairwise distinct except maybe $v_0,v_s$.

%\medskip
%\noindent
%{\bf Partitions of integers.}
%For a positive integer $n$, a sequence of positive integers $(d_1,\ldots,d_s)$, $d_1\leq\ldots\leq d_s$, is a \emph{partition} of $n$ if $n=d_1+\ldots+d_s$. 
% By $p(n)$  we denote the \emph{partition function}, i.e., the number of  partitions of $n$. To obtain an upper bound for $p(n)$, we can use the asymptotic formula obtained by Hardy and Ramanujan in 1918 and independently by  Uspensky in 1920 (see, e.g., the book of Andrews~\cite{Andrews98}):
%$$p(n)\sim \frac{1}{4\sqrt{3}n}e^{\pi\sqrt{2n/3}}.$$

\medskip
\noindent
{\bf Parameterized Complexity.}
Parameterized complexity is a two dimensional framework
for studying the computational complexity of a problem. One dimension is the input size
$n$ and another one is a parameter $k$. It is said that a problem is \emph{fixed parameter tractable} (or \classFPT), if it can be solved in time $f(k)\cdot n^{O(1)}$ for some function $f$.
A \emph{kernelization} for a parameterized problem is a polynomial algorithm that maps each instance $(x,k)$ with the input $x$ and the parameter $k$ to an instance $(x',k')$ such that i) $(x,k)$ is a YES-instance if and only if $(x',k')$ is a YES-instance of the problem, and ii) the size of $x'$ is bounded by $f(k)$ for a computable function $f$. 
The output $(x',k')$ is called a \emph{kernel}. The function $f$ is said to be a \emph{size} of a kernel. Respectively, a kernel is \emph{polynomial} if $f$ is polynomial. 
We refer to the books of Downey and Fellows~\cite{DowneyF99}, 
Flum and Grohe~\cite{FlumG06}, and   Niedermeier~\cite{Niedermeierbook06} for  detailed introductions  to parameterized complexity. 

\medskip
\noindent
{\bf Solutions of \textsc{Editing to a Graph of Given Degrees}.}
Let $(G,\delta,d,k)$ be an instance of \textsc{Editing to a Graph of Given Degrees}.
Let $U\subset V(G)$, $D\subseteq E(G-U)$ and $A\subseteq \binom{V(G)\setminus U}{2}$. If the vertex deletion, edge deletion or edge addition is not in $S$, then it is assumed that $U=\emptyset$, $D=\emptyset$ or $A=\emptyset$ respectively. We say that $(U,D,A)$ is a \emph{solution} for $(G,\delta,d,k)$, if $|U|+|D|+|A|\leq k$, and for the graph $G'=G-U-D+A$, $d_{G'}(v)=\delta(v)$ for $v\in V(G')$.  
We also say that $G'$ is obtained by editing with respect to $(U,D,A)$.

\section{\classFPT-algorithm for Editing to a Graph of Given Degrees}\label{sec:FPT}
Throughout this section we assume that $S=\{\text{vertex deletion},\text{edge deletion},\\ \text{edge addition}\}$, i.e., the all three editing operations are allowed, unless we explicitly specify the set of allowed operations. 
We prove the following theorem.

\begin{theorem}\label{thm:fpt}
{\sc Editing to a Graph of Given Degrees} can be solved in time $2^{O(kd^2+k\log k)}\cdot poly(n)$ for $n$-vertex graphs.
\end{theorem}

\subsection{Preliminaries}
We need the following corollary of the results by Mathieson and Szeider in~\cite{MathiesonS12}.

\begin{lemma}\label{lem:brute}
{\sc Editing to a Graph of Given Degrees} can be solved in time $O^*(2^n)$ for $n$-vertex graphs.
\end{lemma}

\begin{proof}
 Mathieson and Szeider in~\cite{MathiesonS12} proved that {\sc Editing to a Graph of Given Degrees} can be solved in polynomial time if only edge deletion and edge additions are allowed. Since the set of deleted vertices of a hypothetical solution $(U,D,A)$ of {\sc Editing to a Graph of Given Degrees} can be guessed by brute force and we have at most $2^n$ possibilities to choose this set, we can reduce
{\sc Editing to a Graph of Given Degrees} to the case when only  edge deletion and edge additions are allowed by choosing and deleting  $U$, and then solve the problem in polynomial time.
\end{proof}

We also need some structural results about solutions of {\sc Editing to a Graph of Given Degrees}  when  only edge deletion and edge additions are used.

We say that a solution $(U,D,A)$ of \textsc{Editing to a Graph of Given Degrees} is \emph{minimal} if there is no solution $(U',D',A')\neq(U,D,A)$ such that $U'\subseteq U$, $D'\subseteq D$ and $A'\subseteq A$. 

Let $(G,\delta,d,k)$ be an instance of \textsc{Editing to a Graph of Given Degrees}  such that for every $v\in V(G)$, $d_G(v)\leq \delta(v)$.
Let also $(U,D,A)$ be a solution for $(G,\delta,d,k)$ such that $U=\emptyset$, and let 
$G'=G-D+A$. We say that a trail $P=v_0,e_1,v_1,e_2,\ldots,e_s,v_s$  in $G'$ is \emph{$(D,A)$-alternating} if $e_1,\ldots,e_s\subseteq D\cup A$, and
for any $i\in\{2,\ldots,s\}$, either $e_{i-1}\in D,e_i\in A$ or $e_{i-1}\in A,e_i\in D$.
We also say that $P$ is a \emph{degree increasing} trail if  $e_1,e_s\in A$. 
%We say that a $(D,A)$-alternating trail $P=v_0,e_1,v_1,e_2,\ldots,e_s,v_s$ is \emph{closed}, if $v_0=v_s$.
Let $H(D,A)$ be the graph with the edge set $D\cup A$, and 
the vertex set of $H$ consists of the vertices of $G$ incident to the edges of $D\cup A$.
%Let also $Z=\{v\in V(G)|d_G(v)\neq \delta(v)\}$.

\begin{lemma}\label{lem:alt}
Let $(G,\delta,d,k)$ be an instance of \textsc{Editing to a Graph of Given Degrees}  such that for every $v\in V(G)$, $d_G(v)\leq \delta(v)$, and let $Z=\{v\in V(G)|d_G(v)\neq \delta(v)\}$.
For any minimal solution $(U,D,A)$ for $(G,\delta,d,k)$ such that $U=\emptyset$, the graph $H(D,A)$ can be covered by a family of edge-disjoint degree increasing $(D,A)$-alternating 
trails $\mathcal{T}$ (i.e., each edge of $D\cup A$ is in the unique trail of $\mathcal{T}$) with their end-vertices in $Z$. 
\end{lemma}

\begin{proof}
Observe that because for $v\in V(G)\setminus Z$, $d_G(v)=\delta(v)$, 
we have that $|\{e\in D|e\text{ is incident to }v\}|=|\{e\in A|e\text{ is incident to }v\}|$ 
for each $v\in V(H(D,A))\setminus Z$. It implies that $H(D,A)$ can be covered by a family of of edge-disjoint $(D,A)$-alternating 
trails $\mathcal{T}$ such that for every vertex of $v\in V(H(D,A))\setminus Z$, each trails enters $v$ exactly the same number times as it leaves $v$. 
Assume that $\mathcal{T}$ is chosen in such  a way that the number of trails is minimum.
If $\mathcal{T}$ contains a trail $P$ such that $V(P)\cap Z=\emptyset$, then $P$ has even length, and if we delete the edges of $P$ from $D$ and $A$ respectively, we obtain another solution for $(G,\delta,d,k)$, but this contradicts the minimality of $(U,D,A)$. Hence, we can assume that each trail has its end-vertices in $Z$. Suppose that for $v\in Z$, there is a trail $P\in\mathcal{T}$ such that the first or last edge $e$ of $P$ is incident to $v$ and $e\in D$. Because  $d_G(v)< \delta(v)$, there is another trail $P'\in\mathcal{T}$ such that $P'$ starts or ends in $v$, and respectively the first or lase edge $e'$ of $P$ is in $A$. If $P=P\rq{}$, then again we have that 
$P$ has even length, and the deletion of the edges of $P$ from $D$ and $A$ gives another solution for $(G,\delta,d,k)$ contradicting the minimality of $(U,D,A)$. We have that $P\neq P\rq{}$, but 
then we replace $P$ and $P'$ in $\mathcal{T}$ by their concatenation via $v$ and cover 
$H(D,A)$ by $|\mathcal{T}|-1$ paths contradicting the minimality of $\mathcal{T}$. Therefore, 
for each $P\in\mathcal{T}$, the first and last edges of $P$ are in $A$, i.e., $P$ is a degree increasing $(D,A)$-alternating 
trail.
\end{proof}

Using this lemma we obtain the following structural result.

\begin{lemma}\label{lem:struct}
Let $(G,\delta,d,k)$ be an instance of \textsc{Editing to a Graph of Given Degrees}  such that for every $v\in V(G)$, $d_G(v)\leq \delta(v)$, and let $Z=\{v\in V(G)|d_G(v)\neq \delta(v)\}$.
Suppose that $G$ has $r=\lfloor\frac{k}{3}\rfloor$ distinct edges $x_1y_1,\ldots,x_ry_r$ that form a matching such that all  $x_1,\ldots,x_r$ and $y_1,\ldots,y_r$ are distinct  from the vertices of $Z$ and not adjacent to the vertices of $Z$. 
If there is a solution for $(G,\delta,d,k)$ with the empty set of deleted vertices, then the instance has a solution $(U,D,A)$ such that 
\begin{itemize}
\item[i)] $U=\emptyset$,
\item[ii)] either $D=\emptyset$ or $D=\{x_1y_1,\ldots,x_hy_h\}$ for some $h\in\{1,\ldots, r\}$,
\item[iii)] for every $uv\in A$, either $u,v\in Z$ or $uv$ joins $Z$ with some vertex of $\{x_1,\ldots,x_h\}\cup\{y_1,\ldots,y_h\}$,
\item[iv)] for every $i\in\{1,\ldots,h\}$, $A$ has the unique edges $ux_i,vy_i$ such that $u,v\in Z$. 
\end{itemize} 
\end{lemma}

\begin{proof}
Consider a minimal solution $(U,D,A)$ for $(G,\delta,d,k)$ such that $U=\emptyset$. By Lemma~\ref{lem:alt}, $H(D,A)$ can be covered by a family of edge-disjoint degree increasing $(D,A)$-alternating 
trails $\mathcal{T}$ with their end-vertices in $Z$. Let $P_1,\ldots,P_h$ be the trails that have at least one edge from $D$. Because each $P_i$ has at least three edges, $h\leq r$. For each $P_i$,
denote by $u_i,v_i\in Z$ its end-vertices. For $i\in \{1,\ldots,h\}$, we replace $P_i$ by $u_ix_iy_iv_i$. Notice that $u_ix_i,y_iv_i\notin E(G)$ and $x_iy_i\in E(G)$. Respectively, we replace the edges of $P_i$ in $A$ by  $u_ix_i,y_iv_i$, and the edges of $P_i$ in $D$ by $x_iy_i$. It remains to observe that this replacement gives us the solution that satisfies i)-iv).
\end{proof}

\subsection{The algorithm}
We construct an \classFPT-algorithm for \textsc{Editing to a Graph of Given Degrees} parameterized by $k+d$.
The algorithm is based  on the random separation techniques introduced by Cai, Chan and Chan~\cite{CaiCC06} (see also~\cite{AlonYZ95}).

Let $(G,\delta,d,k)$ be an instance of \textsc{Editing to a Graph of Given Degrees}, and let $n=|V(G)|$.%, $m=|E(G)|$. 

\medskip
\noindent
{\bf Preprocessing.} At this stage of the algorithm our main goal is to reduce the original instance of the problem to a bounded number of instances with the property that for any vertex $v$, the degree of $v$ is at most $\delta(v)$.

First , we make the following observation.

\begin{lemma}\label{lem:del}
Let $(U,D,A)$ be a solution for $(G,\delta,d,k)$.
If $d_G(v)>\delta(v)+k$ for $v\in V(G)$, then $v\in U$.
\end{lemma}

\begin{proof}
Suppose that $v\notin U$. Then to obtain a graph $G'$ with $d_{G'}(v)=\delta$, for at least $k+1$ neighbors $u$ of $v$, we should either delete $u$ or delete $uv$. Because the number of editing operations is at most $k$, we immediately obtain a contradiction that proves the lemma.  
\end{proof}

By Lemmas~\ref{lem:del}, we apply the following rule.

\medskip
\noindent
{\bf Vertex deletion rule.} 
If $G$ has a vertex $v$ with $d_G(v)> \delta(v)+k$, then delete $v$ and set $k=k-1$. If $k<0$, then stop and return a NO-answer.

\medskip
We exhaustively apply the rule until we either stop and return a NO-answer or obtain an instance of the problem such that the degree of any vertex $v$ is at most $\delta(v)+k$. 
In the last case it is sufficient to solve the problem for the obtained instance, and if it has a solution $(U,D,A)$, then the solution for the initial instance can be obtained by adding the deleted vertices to $U$.  From now we assume that we do not stop while applying the rule, and to simplify notations, assume that $(G,\delta,d,k)$ is the obtained instance.
Notice that for any $v\in V(G)$, $d_G(v)\leq \delta(v)+k\leq d+k$. Suppose that $v\in V(G)$ and $d_G(v)>\delta(v)$. Then if the considered instance has a solution, either $v$ or at least one of its neighbors should be deleted or at least one of incident to $v$ edges have to be deleted. It implies that we can branch as follows.

\medskip
\noindent
{\bf Branching rule.} If $G$ has a vertex $v$ with $d_G(v)> \delta(v)$, then stop and return a NO-answer if $k=0$, otherwise branch as follows.
\begin{itemize}
\item For each $u\in N_G[v]$, solve the problem for $(G-u,\delta,d,k-1)$, and if there is a solution $(U,D,A)$, then stop and return $(U\cup\{u\},D,A)$.
\item For each $u\in N_G(v)$, solve the problem for $(G-uv,\delta,d,k-1)$, and if there is a solution $(U,D,A)$, then stop and return $(U,D\cup\{uv\},A)$.
\end{itemize}
If none of the instances have a solution,  then return a NO-answer.

\medskip
It is straightforward to observe that by the exhaustive application of 
the rule we either solve the problem or obtain at most $(2(k+d)+1)^k$ instances of the problem such that the original instance has a solution if and only if one of the new instances has a solution, and for each of the obtained instances, the degree of any vertex $v$ is upper bounded by $\delta(v)$.  Now it is sufficient to explain how to solve \textsc{Editing to a Graph of Given Degrees} for such instances. 

To simplify notations, from now we assume that for $(G,\delta,d,k)$, $d_G(v)\leq \delta(v)$ for $v\in V(G)$.
Let $Z=\{v\in V(G)|d_G(v)<\delta(v)\}$. 
%We assume that $Z\neq\emptyset$ because otherwise the problem has the trivial solution $(\emptyset,\emptyset,\emptyset)$.
%In particular, it implies that $d\geq 1$. Also we assume that $k\geq 1$. 
Before we move to the next stage of the algorithm, we use the following lemmas.

\begin{lemma}\label{lem:stop}
If $|Z|>2k$, then the instance $(G,\delta,d,k)$ has no solution.
%If $\sum_{v\in V(G)}(\delta(v)-d_G(v))>2k$,  then the instance $(G,\delta,d,k)$ has no solution.
\end{lemma}

\begin{proof}
Suppose that $(G,\delta,d,k)$ has a solution $(U,D,A)$.
If $d_G(v)<\delta(v)$ for a vertex $v$, then either $v\in U$ or $vu\in A$ for some $u\in V(G)$. It follows that $|Z|\leq |U|+2|A|\leq 2k$.
\end{proof}

\begin{lemma}\label{lem:isol}
Let $n\geq 2$. If $v\in V(G)$ and $d_G(v)=\delta(v)=0$, then $(G-v,\delta,d,k)$ has a solution if and only if $(G,\delta,d,k)$ has a solution, and any solution for  $(G-v,\delta,d,k)$ is a solution for $(G,\delta,d,k)$.
\end{lemma}

\begin{proof}
Let $d_G(v)=\delta(v)=0$. It is straightforward to see that if $(U,D,A)$ is a solution for $(G-v,\delta,d,k)$, then it is a solution for $(G,\delta,d,k)$. Suppose that $(U,D,A)$ is a solution for $(G,\delta,d,k)$. Because $v$ is isolated and $d_G(v)=\delta(v)$,  the vertex $v$ is not incident to any edge of $D$, and we can assume that $v\notin U$ as otherwise $(U\setminus\{v\}, D,A)$ is a solution for  $(G,\delta,d,k)$ as well.
Notice that no edge of $A$ is incident to $v$, because otherwise some edge of $D$ should be incident to $v$ since $d_G(v)=\delta(v)$. 
We conclude that $(U,D,A)$ is a solution for $(G-v,\delta,d,k)$.
\end{proof}

Using Lemma~\ref{lem:stop} and straightforward observations, we apply the following rule.

\medskip
\noindent
{\bf Stopping rule.} If $|Z|>2k$, then stop and return a NO-answer. If $Z=\emptyset$, then stop and return the trivial solution $(\emptyset,\emptyset,\emptyset)$. If $Z\neq\emptyset$ and $k=0$, then stop and return a NO-answer.

\medskip
Then we exhaustively apply the next rule.

\medskip
\noindent
{\bf Isolates removing rule.}
If $G$ has a vertex $v$ with $d_G(v)=\delta(v)=0$, then delete $v$. %If the obtained graph has no vertices, then stop and return the solution $(\emptyset,\emptyset,\emptyset)$.

\medskip
Finally on this stage, we solve small instances.

\medskip
\noindent
{\bf Small instance rule.}
If $G$ has at most $3kd^2-1$ edges, then solve  \textsc{Editing to a Graph of Given Degrees}  using Lemma~\ref{lem:brute}; notice that after the exhaustive application of the previous rule, $G$ has at most $|Z|\leq 2k$ isolated vertices, i.e., $G$ has at most  $6kd^2-2+2k$ vertices.

\medskip
From now we assume that we do not stop at this stage of the algorithm and, as before, denote by $(G,\delta,d,k)$ the obtained instance and assume that $n=|V(G)|$.
 % and $m=|E(G)|$.  
%Let $Z=\{v\in V(G)|d_G(v)<\delta(v)\}$. 
We have that $G$ has at least $3kd^2$ edges,  $|Z|\leq 2k$, $Z\neq\emptyset$, $k\geq 1$,
and for any isolated vertex $v$, $\delta(v)\neq 0$, i.e., $v\in Z$.  Notice that since $Z\neq \emptyset$, $d\geq 1$.

\medskip
\noindent
{\bf Random separation.} Now we apply the random separation technique. We start with constructing a true-biased Monte-Carlo  algorithm and then explain how it can be derandomized. 

We color the vertices of $G$ independently and uniformly at random by two colors. In other words, we partition $V(G)$ into two sets $R$ and $B$. We say that the vertices of $R$ are \emph{red}, and the vertices of $B$ are \emph{blue}.

Let $P=v_0,\ldots,v_s$ be a walk in $G$. We say that $P$ is an \emph{$R$-connecting} walk if either $s\leq 1$ or 
for any $i\in\{0,\ldots,s-2\}$, $\{v_i,v_{i+1},v_{i+2}\}\cap R\neq\emptyset$, i.e., for any three consecutive  vertices of $P$, at least one of them is red.   
We also say that two vertices $x,y$ are \emph{$R$-equivalent} if there is an $R$-connecting walk that joins them. Clearly, $R$-equivalence is an equivalence relation on $R$. Therefore, it defines the corresponding partition of $R$ into equivalence classes. 
Denote by $R_0$ the set of red vertices that can be joined with some vertex of $Z$ by an $R$-connecting walk. Notice that $R_0$ is a union of some equivalence classes.
Denote by $R_1,\ldots,R_t$ the remaining classes, i.e., it is a partition of $R\setminus R_0$ such that any two vertices $x,y$ are in the same set if and only if $x$ and $y$ are $R$-connected; notice that it can happen that $t=0$.
Observe that for any distinct $i,j\in \{0,\ldots,t\}$, $N_G^3[R_i]\cap R_j=\emptyset$ because any two vertices of $R$ at distance at most 3 in $G$ are $R$-equivalent.  
For $i\in\{0,\ldots,t\}$, let $r_i=|R_i|$.
%Let $P=v_0,\ldots,v_s$ be a walk in $G$. We say that $P$ is an \emph{$R$-connecting} walk if  $s\leq 1$ or for any $i\in\{0,\ldots,s-2\}$, $\{v_i,v_{i+1},v_{i+2}\}\cap R\neq\emptyset$, i.e., for any three consecutive  vertices of $P$, at least one of them is red.   Denote by $R_0$ the set of red vertices that can be joined with some vertex of $Z$ by an $R$-connecting walk.  Notice that for any $x,y,z\in R\setminus R_0$, if $x,y$ can be joined by an $R$-connecting walk and  $y,z$ can be joined by an $R$-connecting walk, then $x$ and $z$ can be joined by an $R$-connecting walk. Trivially, any red vertex can be connected with itself by the trivial $R$-connecting walk, and if a red vertex $x$ can be joined with a red vertex $y$ by an $R$-connecting walk, then the same walk joins $y$ with $x$. Hence, any red vertex outside $R_0$ cannot be joined with any red vertex in $R_0$ by an $R$-connecting  walk, because otherwise the vertex should be included in $R_0$.
%Also we can partition the set $R\setminus R_0$ into classes $R_1,\ldots,R_t$ such that 
%for any two vertices $x,y$ in the same set, $x$ and $y$ can be joined by an $R$-connecting walk, and
%for any distinct $i,j\in \{0,\ldots,t\}$, $N_G^3[R_i]\cap R_j=\emptyset$.  For $i\in\{0,\ldots,t\}$, let $r_i=|R_i|$.

The partition $R_0,\ldots,R_t$ can be constructed in polynomial time. To construct $R_0$, we consider the set of red vertices at distance at most two from $Z$ and include them in $R_0$. Then we iteratively include in $R_0$ the red vertices at distance at most two from the vertices included in $R_0$ in the previous iteration. The sets $R_1,\ldots,R_k$ are constructed in a similar way. 

Our aim is to find a solution $(U,D,A)$ for $(G,\delta,d,k)$ such that
\begin{itemize}
\item $U\cap B=\emptyset$,
\item $R_0\subseteq U$,
\item for any $i\in\{1,\ldots,t\}$, either $R_i\subseteq U$ or $R_i\cap U=\emptyset$,
\item the edges of $D$ are not incident to the vertices of $N_G(U)$;
\end{itemize}
i.e., $U$ is a union of equivalence classes of $R$ that contains the vertices of $R_0$. 
We call  such a solution \emph{colorful}. 

Let $B_0=(Z\cap B)\cup N_G(R_0)$, and 
for $i\in\{1,\ldots,t\}$, let $B_i=N_G(R_i)$. Notice that each $B_i\subseteq B$, and for distinct $i,j\in\{0,\ldots,t\}$, the distance between any $u\in B_i$ and $v\in B_j$ is at least two, i.e., $u\neq v$ and $uv\notin E(G)$.   
For a vertex $v\in V(G)\setminus R$, denote by $def(v)=\delta(v)-d_{G-R}(v)$. Recall that $d_G(v)\leq\delta(v)$. Therefore, $def(v)\geq 0$. Notice also that $def(v)$ could be positive only for vertices of the sets $B_0,\ldots,B_t$.  
For each $i\in\{0,\ldots,t\}$, if $v\in B_i$, then either $v\in Z$ or $v$ is adjacent to a vertex of $R_i$. Hence,  $def(v)>0$ for the vertices of $B_0,\ldots,B_t$.
%Because for each $i\in\{1,\ldots,t\}$,
%if $v\in B_i$, then either $v\in Z$ or $v$ is adjacent to a vertex of $R_i$, we have that $def(v)>0$ for the vertices of $B_0,\ldots,B_t$.
For a set $A\subseteq\binom{V(G)}{2}\setminus E(G)$, denote by $d_{G,A}(v)$ the number of elements of $A$ incident to $v$ for $v\in V(G)$.

We construct a dynamic programming algorithm that consecutively for $i=0,\ldots,t$, constructs the table $T_i$
that is either empty, or contains the unique zero element, or contains
lists of all the sequences $(d_1,\ldots,d_p)$ of positive integers, $d_1\leq\ldots\leq d_p$, such that 
\begin{itemize}
\item[i)] there is a set $U\subseteq R_0\cup\ldots\cup R_i$, $R_0\subseteq U$, and 
 for any $j\in\{1,\ldots,i\}$, either $R_j\subseteq U$ or $R_j\cap U=\emptyset$,
\item[ii)] there is a set  $A\subseteq \binom{V(G)}{2}\setminus E(G)$ of pairs of vertices of $B_0\cup\ldots\cup B_i$,  
\item[iii)] $d_1+\ldots+d_p+|U|+|A|\leq k$,
\end{itemize}
and the graph $G'=G-U+A$ has the following properties:
\begin{itemize}
\item[iv)] $d_{G'}(v)\leq \delta(v)$ for $v\in V(G')$, and $d_{G'}(v)< \delta(v)$ for exactly $p$ vertices $v=v_1,\ldots,v_p$,
\item[v)]  $\delta(v_j)-d_{G'}(v_j)=d_i$ for $j\in\{1,\ldots,p\}$.
\end{itemize}
For each sequence $(d_1,\ldots,d_p)$, the algorithm also keeps the sets $U,A$ for which i)--v) are fulfilled
and $|U|+|A|$ is minimum.
The table contains the unique zero element if 
\begin{itemize}
\item[vi)] there is a set $U\subseteq R_0\cup\ldots\cup R_i$, $R_0\subseteq U_i$, and 
 for any $j\in\{1,\ldots,i\}$, either $R_j\subseteq U$ or $R_j\cap U=\emptyset$,
\item[vii)] there is a set  $A\subseteq \binom{V(G)}{2}\setminus E(G)$ of pairs of vertices of $B_0\cup\ldots\cup B_i$,  
\item[viii)]  $|U|+|A|\leq k$, and
\item[ix)] for the graph $G'=G-U+A$, $d_{G'}(v)=\delta(v)$ for $v\in V(G')$.
\end{itemize}
For the zero element, the table stores the corresponding sets $U$ and $A$ for which vi)--ix) are fulfilled.

Now we explain how we construct the tables for $i\in\{0,\ldots,t\}$. 

\medskip
\noindent
{\bf Construction of $T_0$.} Initially we set $T_0=\emptyset$. 
If $\sum_{v\in B_0}def(v)>2(k-|R_0|)$, then we stop, i.e., $T_0=\emptyset$.  
Otherwise, we consider the auxiliary graph $H_0=G[B_0]$. For all sets $A\subseteq \binom{V(H_0)}{2}\setminus E(H_0)$ such that 
for any $v\in V(H_0)$, $d_{H_0,A}(v)\leq def(v)$, and $\sum_{v\in B_0}def(v)-|A|+|R_0|\leq k$, we construct the collection of positive integers 
$Q=\{def(v)-d_{H_0,A}(v)| v\in B_0\text{ and }def(v)-d_{H_0,A}(v)>0\}$ (notice that some elements of $Q$ could be the same). If $Q\neq\emptyset$, then we arrange the elements of $Q$ in increasing order and put the obtained sequence $(d_1,\ldots,d_p)$ of positive integers together with $U=R_0$ and $A$ in $T_0$. If there is $A$ such that $Q=\emptyset$, then we put the zero element in $T_0$ together with $U=R_0$ and $A$, delete all other elements of $T_0$ and then stop, i.e., $T_0$ contains the unique zero element in this case.

\medskip
\noindent
{\bf Construction of $T_i$ for $i\geq 1$.} We assume that $T_{i-1}$ is already constructed. Initially we set $T_i=T_{i-1}$. If $T_i=\emptyset$ or $T_i$ contains the unique zero element, then we stop. Otherwise, we consecutively consider all sequences $(d_1,\ldots,d_p)$ from $T_{i-1}$ with the corresponding sets $U,A$.   
If $\sum_{v\in B_i}def(v)+\sum_{j=1}^pd_i>2(k-|R_i|-|U|-|A|)$, then we stop considering $(d_1,\ldots,d_p)$. 
Otherwise, let $G'=G-U+A$, and let $u_1,\dots,u_p$ be the vertices of $G'$ with $d_j=\delta(u_j)-d_{G'}(u_j)$ for $j\in\{1,\ldots,p\}$.
We consider an auxiliary graph $H_i$ obtained from $G[B_i]$ by the addition of $p$ pairwise adjacent vertices $u_1,\ldots,u_p$. We set $def(u_j)=d_j$ for $j\in\{1,\ldots,p\}$.
For all sets $A'\subseteq \binom{V(H_i)}{2}\setminus E(H_i)$ such that 
for any $v\in V(H_i)$, $d_{H_i,A'}(v)\leq def(v)$, and $\sum_{v\in V(H_i)}def(v)-|A'|+|R_i|+|A|+|U|\leq k$, we construct the collection of positive integers $Q=\{def(v)-d_{H_i,A'}(v)| v\in V(H_i)\text{ and }def(v)-d_{H_i,A'}(v)>0\}$. If $Q\neq\emptyset$, then we arrange the elements of $Q$ in increasing order and obtain the sequence $(d_1',\ldots,d_q')$ of positive integers together with $U''=U\cup R_i$ and $A''=A\cup A'$.
If $(d_1',\ldots,d_q')$ is not in $T_i$, then we add it in $T_i$ together with $U'',A''$. If  $(d_1',\ldots,d_q')$ is already  in $T_i$ together with some sets $U''',A'''$, we replace $U'''$ and $A'''$ by $U''$ and $A''$ respectively if $|U''|+|A''|<|U'''|+|A'''|$.
 If there is $A'$ such that $Q=\emptyset$, then we put the zero element in $T_i$ together with $U''=U\cup U_i$ and $A''=A\cup A'$, delete all other elements of $T_i$ and then stop, i.e., $T_i$ contains the unique zero element in this case. 

The properties of the algorithm are summarized in the following lemma. 
%(the proof is in Appendix~\ref{a:DP}).

\begin{lemma}\label{lem:DP}
The algorithm constructs the tables $T_i$ with at most $2^{O(\sqrt{k})}$ records each in time $2^{O(k\log k)}\cdot poly(n)$ for $i\in\{0,\ldots,t\}$, and each $T_i$ has the following properties.
\begin{enumerate}
\item If $T_i=\emptyset$, then  $(G,\delta,d,k)$ has no colorful solution.
\item The table $T_i$ contains the zero element if and only if $(G,\delta,d,k)$ has a colorful solution $(U,D,A)$ with $R_j\cap U=\emptyset$ for $j\in\{i+1,\ldots,t\}$ and $D=\emptyset$.
Moreover, if  $T_i$ contains the zero element with sets $U,A$, then 
 $(U,\emptyset,A)$ is a colorful solution for $(G,\delta,d,k)$.
\item The table $T_i$ contains a sequence $(d_1,\ldots,d_p)$ of positive integers, $d_1\leq\ldots\leq d_p$, if and only if there are sets $U,A$ that satisfy the conditions i)--v) given above. 
Moreover, if   $T_i$ contains $(d_1,\ldots,d_p)$ together with $U,A$, then $U,A$ are sets with minimum value of $|U|+|A|$ that satisfy i)--v). 
\end{enumerate}
\end{lemma}

\begin{proof}
The proof is inductive. 

First, we show 1)--3) for $i=0$.
Recall that for any colorful solution $(U,D,A)$, $R_0\subseteq U$. 
If $\sum_{v\in B_0}def(v)>2(k-|R_0|)$, then because the addition of any edge increases the degrees of its end-vertices by one, we immediately conclude that there is no colored solution in this case.
Suppose that $(U,D,A)$ is a colorful solution. 
Let $A'=\{uv\in A|u,v\in B_0\}$ and $A''=\{uv\in A|u\in B_0,v\notin B_0\}$. Then $k\geq |U|+|D|+|A|\geq |R_0|+|A'|+|A''|$. Also each edge of $A'$ increases the degrees of two its end-vertices in $B_0$ by one, and each edge  of $A''$ increases the degree of its single end-vertex in $B_0$. Hence,
%$\sum_{v\in B_0}def(v)\leq 2|A'|+|A''|$ and 
$\sum_{v\in B_0}def(v)\leq 2|A'|+|A''|$ and $\sum_{v\in B_0}def(v)-|A'|+|R_0|\leq |A'|+|A''|+|R_0|\leq |A|+|U|\leq k$. Because for the colorful solution $(U,D,A)$, the edges of $D$ are not incident to the vertices of $R_0$, $d_{H_0,A'}(v)\leq def(v)$ for $v\in B_0$. Therefore, if a colorful solution exists, then $T_0\neq\emptyset$. The claims 2 and 3 follows directly from the description of the construction of $T_0$; it is sufficient to observe that we try all possibilities to select the set $A$.  
 
Suppose that $i\geq 1$ and assume that $T_{i-1}$ satisfies 1)--3). Notice that $T_i=\emptyset$ if and only if $T_{i-1}=\ldots=T_0=\emptyset$, and we already proved that if  $T_{0}=\emptyset$, then  $(G,\delta,d,k)$ has no colorful solution.

Now we prove the second claim.

Suppose that $T_{i}$ has the zero element. Then $T_{i-1}\neq\emptyset$. If $T_{i-1}$ has the zero element, then 
by the inductive assumption, $(G,\delta,d,k)$ has a colorful solution $(U,D,A)$ with $R_j\cap U=\emptyset$ for $j\in\{i,\ldots,t\}$ and $D=\emptyset$. Clearly, it is 
 a colorful solution  with $R_j\cap U=\emptyset$ for $j\in\{i+1,\ldots,t\}$ and $D=\emptyset$.
Moreover, if $U,A$ are the sets that are in $T_{i-1}$, then they are in $T_i$ and $(U,\emptyset,A)$ is a colorful solution. 
Suppose now that $T_{i-1}$ does not contain the zero element. Then
there is $(d_1,\ldots,d_p)$ from $T_{i-1}$ with the corresponding sets $U,A$ such that we obtain the zero element when considering this record.   Let $G'=G-U+A$, and let $u_1,\dots,u_p$ be the vertices of $G'$ with $d_j=\delta(u_j)-d_{G'}(u_j)$ for $j\in\{1,\ldots,p\}$.
Also we have a set $A'\subseteq \binom{V(H_i)}{2}\setminus E(H_i)$ such that 
for any $v\in V(H_i)$, $d_{H_i,A}(v)\leq def(v)$,  $\sum_{v\in V(H_i)}def(v)-|A'|+|R_i|+|A|+|U|\leq k$,
and  the collection  of positive integers $Q=\{def(v)-d_{H_i,A'}(v)| v\in V(H_i)\text{ and }def(v)-d_{H_i,A'}(v)>0\}=\emptyset$.
Notice that $u_1,\ldots,u_p$ are at distance at least two from the vertices of $B_i$. Hence, 
$A'\subseteq \binom{V(G')\setminus R_i}{2}$ and $A''=A\cup A'\subseteq\binom{V(G)\setminus(U\cup R_i)}{2}$. Let $G''=G-U''+A''$ where $U''=U\cup R_i$.
By the construction, $d_{G''}(v)\leq \delta(v)$ for $v\in V(G'')$.
Observe that because $Q$ is empty, $\sum_{v\in V(H_i)}def(v)= 2|A'|$. 
Then, $|U''|+|A''|=|A'|+|R_i|+|A|+|U|= \sum_{v\in V(H_i)}def(v)-|A'| +|R_i|+|A|+|U|\leq k$.
We conclude that $(U'',\emptyset,A'')$ is a solution for $(G,\delta,d,k)$. By the construction,
$R_j\cap U''=\emptyset$ for $j\in\{i+1,\ldots,t\}$.

Suppose that $(G,\delta,d,k)$ has a colorful solution $(U,\emptyset,A)$ with $R_j\cap U=\emptyset$ for $j\in\{i,\ldots,t\}$ and $D=\emptyset$. Then by the inductive assumption, $T_{i-1}$ has the zero element, and we have that $T_i$ contains the same element. Suppose now that $(G,\delta,d,k)$ has no such a solution, but it has  a colorful solution $(U'',\emptyset,A'')$ with $R_j\cap U''=\emptyset$ for $j\in\{i+1,\ldots,t\}$. Then $R_i\subseteq U''$.  Also we have that $R_0\subseteq U''\setminus R_i\subseteq R_0\cup\ldots\cup R_{i-1}$. Consider the partition $A_1,A_2,A_3$ (some sets can be empty) of $A''$ such that the edges of $A_1$ join vertices $B_i$, the edges of $A_2$ join $B_i$ with vertices of $B_0\cup\ldots\cup B_{i-1}$, and the edges of $A_3$ join verices of $B_0\cup\ldots\cup B_{i-1}$. Let $w_1,\ldots,w_p$ be the end-vertices of the edges of $A_2$ in $B_0\cup\ldots\cup B_{i-1}$, 
$d_j=d_{F,A_2}(w_j)$ for $j\in\{1,\ldots,p\}$ where $F=G-U''$, and assume that $d_1\leq\ldots\leq d_p$.
Consider $F'=G-(U''\setminus R_i)+A_3$. Then $\{w_1,\ldots,w_p\}=\{v\in V(F')|\delta(v)>d_{F'}(v)\}$ and  $d_j=\delta(w_j)-d_{F'}(w_j)$
for $j\in\{1,\ldots,p\}$. Therefore, $T_{i-1}$ contains the record with the sequence $(d_1,\ldots,d_p)$ and some sets $U,A$. 
Let $G'=G-U+A$ and let $u_1,\dots,u_p$ be the vertices of $G'$ with $d_j=\delta(u_j)-d_{G'}(u_j)$ for $j\in\{1,\ldots,p\}$.
Notice that $u_1,\ldots,u_p$ are at distance at least two from the vertices of $B_i$. We construct $A_2'$ by replacing each edge $vw_j$ by $vu_j$ for $j\in\{1,\ldots,p\}$.
Let $A'=A_1\cup A_2'$.
We have that  $A'\subseteq \binom{V(H_i)}{2}\setminus E(H_i)$,  
for any $v\in V(H_i)$, $d_{H_i,A'}(v)\leq def(v)$, and 
$\sum_{v\in V(H_i)}def(v)-2|A'|=0$.
Also because $U,A$ are chosen in such a way that $|U|+|A|$ has minimum size for $(d_1,\ldots,d_p)$,
$\sum_{v\in V(H_i)}def(v)-|A'|+|R_i|+|A|+|U|=\sum_{v\in V(H_i)}def(v)-2|A'|+|A_1|+|A_2|+|A|+|R_i|+|U|\leq |A_1|+|A_2|+|A_3|+|U''|\leq k$.
Then for $A'$, we construct $Q=\{def(v)-d_{H_i,A'}(v)| v\in V(H_i)\text{ and }def(v)-d_{H_i,A'}(v)>0\}$, and because $Q=\emptyset$,  we put the zero element in $T_0$ together with $U\cup U_i$ and $A\cup A'$, delete all other elements of $T_i$ and then stop, i.e., $T_i$ contains the unique zero element in this case. 

The third claim is proved by similar arguments.

 Suppose $T_i$ contains $(d_1',\ldots,d_q')$. % together with $U,A$.  
If $(d_1',\ldots,d_q')$ is in $T_{i-1}$, then by the inductive assumption, there are sets $U'',A''$ that satisfy the conditions i)--v) given above. Suppose that $(d_1',\ldots,d_q')$ is not in $T_{i-1}$. Then there is a sequence $(d_1,\ldots,d_p)$ in $T_{i-1}$ with the corresponding sets $U,A$ such that we obtain the sequence $(d_1',\ldots,d_q')$ when considering this record.
Also for the graph $H_i$, we have $A'\subseteq \binom{V(H_i)}{2}\setminus E(H_i)$ such that we obtain $(d_1',\ldots,d_q')$ by orderning $Q=\{def(v)-d_{H_i,A'}(v)| v\in V(H_i)\text{ and }def(v)-d_{H_i,A'}(v)>0\}$.
Then it is straightforward to verify that $U''=U\cup R_i$  and $A''=A\cup A'$ satisfy i)--v). 

We have that if $(d_1',\ldots,d_q')$ is in $T_i$ together with $U'',A''$ then i)--v) are fulfilled. Assume that $|U''|+|A''|$ is not minimal, i.e., there are other sets $\hat{U},\hat{A}$ such that $|\hat{U}|+|\hat{A}|<|U''|+|A''|$ and i)--v) are fulfilled for these sets. If $\hat{U}\cap R_i=\emptyset$, then we have that 
$(d_1',\ldots,d_q')$ together with some $\hat{U}',\hat{A}'$ such that $|\hat{U}'|+|\hat{A}'|\leq |\hat{U}|+|\hat{A}|$, and 
$(d_1',\ldots,d_q')$ with $\hat{U}',\hat{A}'$ instead of $U'',A''$ should be in $T_{i-1}$, but the records of $T_{i-1}$ are included in $T_i$ in the beginning and we have a contradiction. Hence, $R_i\subseteq \hat{U}$.
Consider the partition $A_1,A_2,A_3$ (some sets can be empty) of $\hat{A}$ such that the edges of $A_1$ join vertices $B_i$, the edges of $A_2$ join $B_i$ with vertices of $B_0\cup\ldots\cup B_{i-1}$, and the edges of $A_3$ join verices of $B_0\cup\ldots\cup B_{i-1}$. Let $w_1,\ldots,w_p$ be the end-vertices of the edges of $A_2$ in $B_0\cup\ldots\cup B_{i-1}$, 
$d_j=d_{F,A_2}(w_j)$ for $j\in\{1,\ldots,p\}$ where $F=G-\hat{U}$ and assume that $d_1\leq\ldots\leq d_p$.
Consider $F'=G-(\hat{U}\setminus R_i)+A_3$. Then $\{w_1,\ldots,w_p\}=\{v\in V(F')|\delta(v)>d_{F'}(v)\}$ and  $d_j=\delta(w_j)-d_{F'}(w_j)$
for $j\in\{1,\ldots,p\}$. Therefore, $T_{i-1}$ contains the record with the sequence $(d_1,\ldots,d_p)$ and some sets $U,A$. 
Let $G'=G-U+A$, and let $u_1,\dots,u_p$ be the vertices of $G'$ with $d_j=\delta(u_j)-d_{G'}(u_j)$ for $j\in\{1,\ldots,p\}$.
Notice that $u_1,\ldots,u_p$ are at distance at least two from the vertices of $B_i$. We construct $A_2'$ by replacing each edge $vw_j$ by $vu_j$ for $j\in\{1,\ldots,p\}$.
Let $A'=A_1\cup A_2'$.
We have that  $A'\subseteq \binom{V(H_i)}{2}\setminus E(H_i)$,  
for any $v\in V(H_i)$, $d_{H_i,A'}(v)\leq def(v)$, and
$\sum_{v\in V(H_i)}def(v)-2|A'|\geq 0$.
Also because $U,A$ are chosen in such a way that $|U|+|A|$ has minimum size for $(d_1,\ldots,d_p)$,
$\sum_{v\in V(H_i)}def(v)-|A'|+|R_i|+|A|+|U|\leq\sum_{v\in V(H_i)}def(v)-|A_1|-|A_2|+|R_i|+|A_3|+|\hat{U}|\leq k$.
Then for $A'$, we construct $Q=\{def(v)-d_{H_i,A'}(v)| v\in V(H_i)\text{ and }def(v)-d_{H_i,A'}(v)>0\}$, and because $Q=\{d_1',\ldots,d_q'\}$,  we obtain 
$(d_1',\ldots,d_q')$ with $U\cup R_i$ and $A\cup A'$, but since $|U\cup R_i|+|A\cup A'|\leq |\hat{U}|+|\hat{A}|$, we should put these sets in $T_i$ instead of $U'',A''$; a contradiction. 

Suppose now that $T_i$ contains $(d_1',\ldots,d_q')$ with $U'',A''$ that satisfy i)--v). If $R_i\cap U''=\emptyset$, then $(d_1',\ldots,d_q')$ is in $T_{i-1}$ by the inductive assumption. Therefore, the sequence is in $T_i$ as well. Suppose that $R_i\subseteq U''$. 
Consider the partition $A_1,A_2,A_3$ (some sets can be empty) of $A''$ such that the edges of $A_1$ join vertices $B_i$, the edges of $A_2$ join $B_i$ with vertices of $B_0\cup\ldots\cup B_{i-1}$, and the edges of $A_3$ join verices of $B_0\cup\ldots\cup B_{i-1}$. Let $w_1,\ldots,w_p$ be the end-vertices of the edges of $A_2$ in $B_0\cup\ldots\cup B_{i-1}$, 
$d_j=d_{F,A_2}(w_j)$ for $j\in\{1,\ldots,p\}$ where $F=G-\hat{U}$, and assume that $d_1\leq\ldots\leq d_p$.
Consider $F'=G-(\hat{U}\setminus R_i)+A_3$. Then $\{w_1,\ldots,w_p\}=\{v\in V(F')|\delta(v)>d_{F'}(v)\}$ and  $d_j=\delta(w_j)-d_{F'}(w_j)$
for $j\in\{1,\ldots,p\}$. Therefore, $T_{i-1}$ contains the record with the sequence $(d_1,\ldots,d_p)$ and some sets $U,A$. 
Let $G'=G-U+A$, and let $u_1,\dots,u_p$ be the vertices of $G'$ with $d_j=\delta(u_j)-d_{G'}(u_j)$ for $j\in\{1,\ldots,p\}$.
Notice that $u_1,\ldots,u_p$ are at distance at least two from the vertices of $B_i$. We construct $A_2'$ by replacing each edge $vw_j$ by $vu_j$ for $j\in\{1,\ldots,p\}$.
It remains to observe that we include $(d_1',\ldots,d_q')$ in $T_i$ when we consider $(d_1,\ldots,d_p)$ from $T_{i-1}$ and the set $A'=A_1\cup A_2'\subseteq \binom{V(H_i)}{2}\setminus E(H_i)$.  

It remains to obtain the upper bound for the number of elements in each table and evaluate the running time.

For a positive integer $\ell$, a sequence of positive integers $(\ell_1,\ldots,\ell_s)$, $\ell_1\leq\ldots\leq \ell_s$, is a \emph{partition} of $\ell$ if $\ell=\ell_1+\ldots+\ell_s$. 
% By $p(n)$  we denote the \emph{partition function}, i.e., the number of  partitions of $n$. 
To obtain an upper bound for the number of partitions $\pi(\ell)$, we can use the asymptotic formula obtained by Hardy and Ramanujan in 1918 and independently by  Uspensky in 1920 (see, e.g., the book of Andrews~\cite{Andrews98}):
$$\pi(\ell)\sim \frac{1}{4\sqrt{3}\ell}e^{\pi\sqrt{2\ell/3}}.$$
Observe that because for each sequence $(d_1,\ldots,d_p)$ in a table, $d_1+\ldots+d_p\leq k$, the total number of sequences in each table is upper bounded by 
$k\pi(k)$. Hence, by the asymptotic formula of Hardy and Ramanujan, the number of records in each table is $2^{O(\sqrt{k})}$.   

Notice that to construct $T_i$, we consider the graphs $H_i$ that has at most $2k$ vertices. Hence, $H_i$ have at most $2k^2$ pairs of non-adjacent vertices. Among these pairs we choose at most $k$ pairs. Hence, for each $H_i$, we have at most $2^{O(k\log k)}$ possibilities. For $i\geq 1$, we construct $H_i$ for each element of $T_{i-1}$. Therefore, the each table is constructed in time $2^{O(k\log k)}2^{O(\sqrt{k})}\cdot poly(n)$. Because the number of tables is at most $n$, we have that the algorithm runs in time $2^{O(k\log k)}\cdot poly(n)$.  
\end{proof}

We use the final table $T_t$ to find a colorful solution for $(G,\delta,d,k)$ if it exists. 
\begin{itemize}
\item If $T_r$ contains the zero element with $U,A$, then $(U,\emptyset,A)$ is a colorful solution.
\item If $T_r$ contains a sequence $(d_1,\ldots,d_p)$ with $U,A$ such that $3(d_1+\ldots+ d_p)/2+|U|+|A|\leq k$ and $r=d_1+\ldots+d_p$ is even, then let $G'=G-U+A$ and 
find the vertices $u_1,\ldots,u_p$ of $G'$ such that $\delta(u_i)-d_{G'}(u_i)=d_i$ for $i\in\{1,\ldots,p\}$.
Then greedily 
%let $F=G-U$ and
 find a matching $D$ in  %$F$ 
$G'$
with $h=r/2$ edges $x_1y_1,\ldots,x_hy_h$ such that $x_1,\ldots,x_h$ and $y_1,\ldots,y_h$ are distinct  from the vertices of $\{u_1,\ldots,u_p\}\cup N_G(U)$ and not adjacent to $u_1,\ldots,u_p$. 
%the vertices of $Z'=\{v\in V(G')|d_{G'}(v)<\delta(v)\}$ and not adjacent to the vertices of $Z'$. 
Then we construct the set $A'$ as follows. Initially $A'=\emptyset$. Then
for each $i\in\{1,\ldots,r\}$, we consecutively select next $d_i$ vertices $w_1,\ldots,w_{d_i}\in \{x_1,\ldots,x_h,y_1,\ldots,y_h\}$ in such a way that each vertex is selected exactly once and add in  $A'$ the pairs $u_1w_1,\ldots, u_iw_{d_i}$.
Then we output the solution $(U,D,A\cup A')$.
\item In all other cases we have a NO-answer.
\end{itemize}

\begin{lemma}\label{lem:colorful}
The described algorithm finds a colorful solution for $(G,\delta,d,k)$ if it exists, and it returns a NO-answer otherwise.
\end{lemma}

\begin{proof}
If $T_r$ contains the zero element with $U,A$, then $(U,\emptyset,A)$ is a colorful solution by Lemma~\ref{lem:DP}.

Suppose $T_r$ contains a sequence $(d_1,\ldots,d_p)$ with $U,A$ such that $3(d_1+\ldots+ d_p)/2+|U|+|A|\leq k$ and $r=d_1+\ldots+d_p$ is even. 
Observe that  if $D$ and $A'$ exist, then $(U,D,A\cup A')$ is a solution.
The graph $G'=G-U+A$ has $p$ vertices $u_1,\ldots,u_p$  such that $\delta(u_i)-d_{G'}(u_i)=d_i$ for $i\in\{1,\ldots,p\}$, and for any other vertex $v$, $d_{G'}(v)=\delta(v)$.
%Let $G\rq{}\rq{}=G-U$. The graph $G\rq{}\rq{}$ has at most $2(k-|U|)$ vertices $v$ such that  $\delta(v)>d_{G'}(v)$ as otherwise $(d_1,\ldots,d_p)$ with $U,A$ would not be included in $T_t$.
Observe that  $p\leq k-|U|$. Also at most $|U|d^2$ edges are incident to the vertices of $N_G(U)$. 
 Recall that $G$ has  at least $3kd^2$ edges. Therefore, $G'$ has at least $3kd^2-pd^2-|U|d^2\geq 2kd^2$ edges that are not incident to $\{u_1,\ldots,u_p\}\cup N_G(U)$ and the vertices that are adjacent to $u_1,\ldots,u_p$. 
%Because $d_{G'}(v)\leq \delta(v)\leq d$ for $v\in V(G\rq{}\rq{})$, at most $2(k-|U|)d^2$ edges are incident to the vertices  of $\{u_1,\ldots,u_p\}\cup N_G(U)$  and the vertices that are adjacent to them. Also at most $|U|d$ edges are incident to the vertices of $U$ in $G$.  Recall that $G$ has  at least $3kd^2$ edges. Therefore, $G'$ has at least $3kd^2-2(k-|U|)d^2-|U|d\geq kd^2$ edges that are not incident to the vertices of $\{u_1,\ldots,u_p\}\cup N_G(U)$ and the vertices that are adjacent to them. 
Then $h\leq k/3$ edges of $D$ can be selected greedily by the consecutive arbitrary choice of $x_iy_i$  and the deletion of at most $2d-1$ edges incident to $x_i,y_i$. Because $2h=r=d_1+\ldots+d_p$, we always can join $u_1,\ldots,u_p$ with $x_1,\ldots,x_h$, $y_1,\ldots,y_h$ by edges as prescribed.

Now we show that if  $(G,\delta,d,k)$ has a colorful solution $(U,D,A)$, then the algorithm outputs some colorful solution. 

Consider the graph $G'=G-U$ and let $k'=k-|U|$. 
Clearly, $(G',\delta,d,k')$ is an instance of \textsc{Editing to a Graph of Given Degrees}  such that for every $v\in V(G')$, $d_{G'}(v)\leq \delta(v)\leq d$, and it has 
 a solution  with the empty set of deleted vertices. 
Let $Z'=\{v\in V(G')|d_{G'}(v)< \delta(v)\}$. By Lemma~\ref{lem:stop}, $|Z'|\leq 2k'$. Then at most $2k'd^2$ edges of $G'$ are incident to the vertices of $Z'$
and  the vertices that are adjacent to them.
 Also at most $|U|d$ edges are incident to the vertices of $U$ in $G$.  Because $G$ has  at least $3kd^2$ edges,  $G'$ has at least $3kd^2-2(k-|U|)d^2-|U|d\geq kd^2$ edges that are not incident to the vertices of $Z'$ and the vertices that are adjacent to them. Then a matching with at least $\lfloor k/3\rfloor$ edges $x_1y_1,\ldots,x_sy_s$ with their end-vertices at distance at least two from $Z'$ can be selected greedily. By Lemma~\ref{lem:struct}, $(G',\delta,d,k')$ has a solution $(U,D',A')$ such that 
\begin{itemize}
\item[i)] $U=\emptyset$,
\item[ii)] either $D'=\emptyset$ or $D'=\{x_1y_1,\ldots,x_hy_h\}$ for some $h\in\{1,\ldots, s\}$, 
\item[iii)] for every $uv\in A'$, either $u,v\in Z'$ or $uv$ joins $Z'$ with some vertex of $\{x_1,\ldots,x_h\}\cup\{y_1,\ldots,y_h\}$,
\item[iv)] for every $i\in\{1,\ldots,h\}$, $A'$ has the unique edges $ux_i,vy_i$ such that $u,v\in Z$. 
\end{itemize} 
Let $A''=\{uv\in A'|u,v\in Z'\}$.  Consider $G''=G-U+A''$. Let $u_1,\ldots,u_p$ be the vertices of $G''$ such that $d_{G''}(u_i)<\delta(u_i)$ for $i\in\{1,\ldots,p\}$. Let $d_i=\delta(u_i)-d_{G''}(u_i)$ and assume that $d_1\leq\ldots\leq d_p$. 
Notice that $3(d_1+\ldots+ d_p)/2+|U|+|A''|\leq k$. 
We have that $T_t$ contains $(d_1,\ldots,d_p)$ with some sets $U''',A'''$ and
$|U'''|+|A'''|\leq|U|+|A''|$. Then $3(d_1+\ldots+ d_p)/2+|U'''|+|A'''|\leq k$ and the algorithm finds a colorful solution for the instance 
 $(G,\delta,d,k)$.
\end{proof}

The described algorithm finds a colorful solution if it exists. To find a solution, we run the randomized algorithm  $N$ times. If we find a solution after some run, we return it and stop. If we do not obtain a solution after $N$ runs, we return a NO-answer. The next lemma shows that it is sufficient to run the algorithm $N=2^{O(dk^2)}$ times.

\begin{lemma}\label{lem:N-bound}
If after $N=2^{4kd^2}$ executions the randomized algorithm does not find a solution for $(G,\delta,d,k)$, then it does not exists with a positive probability $p$ such that $p$ does not depend on the instance.
\end{lemma}

\begin{proof}
Suppose that $(G,\delta,d,k)$ has a solution $(U,D,A)$.  The algorithm colors the vertices of $G$ independently and uniformly at random by two colors. 

We find a lower bound for the probability that the vertices of $N_G^2[Z]\cup N_G^3[U]$ are colored correctly with respect to the solution, i.e., the vertices of $U$ are red and all other vertices are blue.
Recall that $d_G(v)\leq d$ for $v\in V(G)$ and $d_G(v)\leq d-1$ for $v\in Z$. Recall also that $|Z|\leq 2k$. 
Hence, $|N_G^2[Z]|\leq 2kd^2-4kd+2k\leq 2kd^2-2k$. The set $U$ has at most $k$ vertices. Because for each $v\in N_G(U)$, its degree in $G-U$ is at most $d_G(v)-1<\delta(v)$, by Lemma~\ref{lem:stop}, $|N_G(U)|\leq 2k$. Therefore, $N_G^3[U]\leq 2kd^2+2k$. We have that  $|N_G^2[Z]\cup N_G^3[U]|\leq 4kd^2$. Hence, we color the set correctly with the probability at least
$2^{-4kd^2}$.

Assume that the random coloring colored $N_G^2[Z]\cup N_G^3[U]$ correctly with respect to the solution $(U,D,A)$. 
Recall $R_0$ is the set of red vertices that can be joined with some vertex of $Z$ by an $R$-connecting walk. Because the vertices of $N_G^2[Z]$ and the vertices of $N_G^3[U]$ are colored correctly,  we have that $R_0$ contains only red vertices from $U$. Also for other sets $R_1,\ldots,R_t$ of the partition of $R$, we have that each $R_i\subseteq U$ or $R_i\cap U=\emptyset$.
It follows, that the problem has a colorful solution in this case, and the algorithm finds it.

The probability that  the vertices of $N_G^2[Z]\cup N_G^3[U]$ are not colored correctly with respect to $(U,D,A)$ is at most $(1-2^{-4kd^2})$, and the probability that these vertices are non colored correctly with respect to the solution for neither of $N=2^{4kd^2}$ random colorings is at most   $(1-2^{-4kd^2})^{4kd^2}$, and the claim follows. 
\end{proof}

The algorithm can be derandomized by standard techniques (see~ \cite{AlonYZ95,CaiCC06}) because random colorings can be replaced by the colorings induced by \emph{universal sets}.
Let $n$ and $r$ be positive integers, $r\leq n$. An  \emph{$(n,r)$-universal set} is a collection of binary vectors of length $n$ such that for each index subset of size $r$, each of the $2^r$ possible combinations of values appears in some vector of the set. It is known that an $(n,r)$-universal set can be constructed in \classFPT-time with the parameter $r$. The best construction is due to Naor, Schulman and Srinivasan~\cite{naor1995splitters}. They obtained an $(n,r)$-universal set of size $2^r\cdot r^{O(\log r)} \log n$, and proved that the elements of the sets  can be listed in time that is linear in the size of the set. 

To apply this technique in our case, we construct an $(n,r)$-universal set $\mathcal{U}$ for $r=\min\{4kd^2,n\}$. Then we let $V(G)=\{v_1,\ldots,v_n\}$ and for each element of  $\mathcal{U}$, i.e., a binary vector $x=(x_1,\ldots,x_n)$, we consider the coloring of $G$ induced by $x$; a vertex $v_i$ is colored red if $x_i=1$, and $v_i$ is blue otherwise.   Then if $(G,\delta,d,k)$ has a solution $(U,D,A)$, then for one of these colorings, the vertices of $N_G^2[Z]\cup N_G^3[U]$ are colored correctly with respect to the solution, i.e., the vertices of $U$ are red and all other vertices of the set are blue.
In this case the instance has a colorful solution, and our algorithm finds it.

\medskip
\noindent
{\bf Running time.} We conclude %this section and 
the proof of Theorem~\ref{thm:fpt} by the running time analysis. 

Clearly, the vertex deletion rule can be applied in polynomial time. The branching rule produces at most $(2(k+d)+1)^k$ instances of the problem and can be implemented in time $2^{O(k\log(k+d)}\cdot poly(n)$. Then the stopping and isolates removing rules can be done in polynomial time.   Whenever we apply the small instance rule, we have an instance with the graph with at most $3kd^2-1$ edges with at most $2k$ isolated vertices. Hence, the graph has at most $2(3kd^2-1)+2k$ vertices. By Lemma~\ref{lem:brute}, the problem can be solved in time $2^{O(kd^2)}\cdot poly(n)$. Hence, on the preprocessing stage we either solve the problem or 
produce at most $(2(k+d)+1)^k$ new instances of the problem in time $2^{O(kd^2+k\log k)}\cdot poly(n)$.

For each coloring of $G$, we can construct the partition $R_0,\ldots,R_t$ of $R$ in polynomial time. Then the dynamic programming algorithm produces the table $T_t$ in time $2^{O(k\log k)}\cdot poly(n)$. Using the information in $T_t$, we solve the problem in time  $2^{O(\sqrt{k})}\cdot poly(n)$ because $T_t$ has at most $2^{O(\sqrt{k})}$ records. Hence, for each coloring the problem is solved time   $2^{O(k\log k)}\cdot poly(n)$. We either consider at most  $N=2^{4kd^2}$ random colorings or at most $2^r\cdot r^{O(\log r)} \log n$ elements of an $(n,r)$-universal set for $r\leq 4kd^2$. In the both cases we have that we can solve the problem in time $2^{O(kd^2+k\log k)}\cdot poly(n)$. Since we solve the problem for at most $(2(k+d)+1)^k$ instances obtained on the preprocessing stage, we have that the total running time is $2^{O(kd^2+k\log k)}\cdot poly(n)$.

\subsection{The case $S=\{\text{vertex deletion},\text{edge addition}\}$}
We conclude the section by the observation that a simplified variant of our algorithm solves {\sc Editing to a Graph of Given Degrees} for $S=\{\text{vertex deletion},\text{edge addition}\}$.
We have to modify the branching rule to exclude edge deletions. Also on the preprocessing stage we don\rq{}t need the small instance rule. On the random separation stage, we simplify the algorithm by the observation that we have a colorful solution if and only if the table $T_t$ has the zero element. It gives us the following corollary.

\begin{corollary}\label{cor:fpt}
{\sc Editing to a Graph of Given Degrees} can be solved in time $2^{O(kd^2+k\log k)}\cdot poly(n)$ for $n$-vertex graphs for $S=\{\text{vertex deletion},\text{edge addition}\}$.
\end{corollary}

\section{Kernelization lower bound for Editing to a Graph of Given Degrees}\label{sec:no-ker}
In this section we show that it is unlikely that \textsc{Editing to a Graph of Given Degrees}  parameterized by $k+d$ has a polynomial kernel if  $\{\text{vertex deletion},\text{edge addition}\}\subseteq S$. The proof uses the cross-composition technique introduced by Bodlaender, Jansen and Kratsch~\cite{BodlaenderJK14}.
We need the following definitions (see~\cite{BodlaenderJK14}).

Let $\Sigma$ be a finite alphabet. An equivalence relation $\mathcal{R}$ on the set of strings $\Sigma^*$ is called a \emph{polynomial equivalence relation} if the following two conditions hold:
\begin{itemize}
\item[i)] there is an algorithm that given two strings $x,y\in\Sigma^*$ decides whether $x$ and $y$ belong to
the same equivalence class in time polynomial in $|x|+|y|$,
\item[ii)] for any finite set $S\subseteq\Sigma^*$, the equivalence relation $\mathcal{R}$ partitions the elements of $S$ into a
number of classes that is polynomially bounded in the size of the largest element of $S$.
\end{itemize}

Let $L\subseteq\Sigma^*$ be a language, let $\mathcal{R}$ be a polynomial
equivalence relation on $\Sigma^*$, and let $\mathcal{Q}\subseteq\Sigma^*\times\mathbb{N}$   
be a parameterized problem.  An \emph{OR-cross-composition of $L$ into $\mathcal{Q}$} (with respect to $\mathcal{R}$) is an algorithm that, given $t$ instances $x_1,x_2,\ldots,x_t\in\Sigma^*$ 
of $L$ belonging to the same equivalence class of $\mathcal{R}$, takes time polynomial in
$\sum_{i=1}^t|x_i|$ and outputs an instance $(y,k)\in \Sigma^*\times \mathbb{N}$ such that:
\begin{itemize}
\item[i)] the parameter value $k$ is polynomially bounded in $\max\{|x_1|,\ldots,|x_t|\} + \log t$,
\item[ii)] the instance $(y,k)$ is a YES-instance for $\mathcal{Q}$ if and only if at least one instance $x_i$ is a YES-instance for $L$ for $i\in\{1,\ldots,t\}$.
\end{itemize}
It is said that $L$ \emph{OR-cross-composes into} $\mathcal{Q}$ if a cross-composition
algorithm exists for a suitable relation $\mathcal{R}$.

In particular, Bodlaender, Jansen and Kratsch~\cite{BodlaenderJK14} proved the following theorem.

\begin{theorem}[\cite{BodlaenderJK14}]\label{thm:BJK}
If an \classNP-hard language $L$ OR-cross-composes into the parameterized problem $\mathcal{Q}$,
then $\mathcal{Q}$ does not admit a polynomial kernelization unless
$\classNP\subseteq\classCoNP/\text{\rm poly}$.
\end{theorem}

It is well-known that the \textsc{Clique} problem is \classNP-complete for regular graphs~\cite{GareyJ79}. We need a special variant of \textsc{Clique} for regular graphs where a required clique is small with respect to the degree.
\begin{center}
\begin{boxedminipage}{.99\textwidth}
\textsc{Small Clique in a Regular Graph}\\
\begin{tabular}{ r l }
\textit{~~~~Instance:} & Positive integers $d$ and $k$, $k\geq 2$, $k^2<d$, and
                                  a $d$-regular\\ & graph $G$.\\
\textit{Question:} & Is there a clique with $k$ vertices in $G$?\\
\end{tabular}
\end{boxedminipage}
\end{center}

\begin{lemma}\label{lem:clique}
\textsc{Small Clique in a Regular Graph} is \classNP-complete.
\end{lemma}

\begin{proof}
Recall that the \textsc{Clique} problem asks for a graph $G$ and a positive integer $k$, whether $G$ has a clique with $k$ vertices.  \textsc{Clique} is known to be \classNP-complete for regular graphs~\cite{GareyJ79} (it can be observed, e.g., that the dual \textsc{Independent Set} problem is \classNP-complete for cubic graphs).  
To show \classNP-hardness of \textsc{Small Clique in a Regular Graph}, we reduce from \textsc{Clique} for regular graphs.

Recall that the Cartesian product of graphs $G$ and $H$ is the graph $G\times H$ with the vertex set $V(G)\times V(H)$ such that $(u,v),(u\rq{},v\rq{})\in V(G)\times V(H)$ are adjacent in $G\times H$ if and only if $u=u\rq{}$ and $vv\rq{}\in E(H)$ or $v=v\rq{}$ and $uu\rq{}\in E(G)$.

Let $G$ be a $d$-regular graph, $d\geq 1$,  and let $k$ be a positive integer. We construct $H=G\times K_{k^2,k^2}$. For any $v\in V(H)$, $d_H(v)=d+k^2$, i.e., $H$ is a $d\rq{}=d+k^2$-regular graph and $d\rq{}>k^2$. It remains to observe that $H$ has a clique of size $k$ if and only if $H$ has a clique of size $k$.
\end{proof}

Now we are ready to prove the main result of the section.

\begin{theorem}\label{thm:no-ker}
\textsc{Editing to a Graph of Given Degrees}  parameterized by $k+d$ has no polynomial kernel unless  $\classNP\subseteq\classCoNP/\text{\rm poly}$ if $\{\text{vertex deletion},\text{edge addition}\}\subseteq S$.
\end{theorem}

\begin{proof}
 First, we consider the case when the all three editing operations are allowed, i.e., $S=\{\text{vertex deletion},\text{edge deletion},\text{edge addition}\}$.

We show that \textsc{Small Clique in a Regular Graph} OR-cross-composes into \textsc{Editing to a Graph of Given Degrees}. 

We say that that two instances $(G_1,d_1,k_1)$ and $(G_2,d_2,k_2)$ of \textsc{Small Clique in a Regular Graph} are \emph{equivalent} if $|V(G_1)|=|V(G_2)|$, $d_1=d_2$ and $k_1=k_2$. Notice that this is a polynomial equivalence relation.

Let $(G_1,d,k),\ldots,(G_t,d,k)$ be equivalent instances of \textsc{Small Clique in a Regular Graph}, $n=|V(G_i)|$ for $i\in\{1,\ldots,t\}$. We construct the instance 
$(G\rq{},\delta,d\rq{},k\rq{})$ of \textsc{Editing to a Graph of Given Degrees} as follows. 
\begin{itemize}
\item Construct copies of $G_1,\ldots,G_t$.
\item Construct $p=k(d-k+1)$ pairwise adjacent vertices $u_1,\ldots,u_p$.
\item Construct $k+1$ pairwise adjacent vertices $w_0,\ldots,w_k$ and join each $w_j$ with each $u_h$ by an edge. 
\item Set $\delta(v)=d$ for $v\in V(G_1)\cup\ldots\cup V(G_t)$, $\delta(u_i)=p+k+1$ for $i\in\{1,\ldots,p\}$, and $\delta(w_j)=p+k$ for $j\in\{0,\ldots,k\}$.
\item Set $d\rq{}=p+k+1$ and $k\rq{}=k(d-k+2)$.
\end{itemize}
Denote the obtained graph $G\rq{}$.

Clearly, $k\rq{}+d\rq{}=O(n^2)$, i.e., the parameter value  is polynomially bounded  in $n$. We show that  $(G\rq{},\delta,d\rq{},k\rq{})$ is a YES-instance of \textsc{Editing to a Graph of Given Degrees} if and only if $(G_i,k,d)$ is a YES-instance of  \textsc{Small Clique in a Regular Graph} for some $i\in\{1,\ldots,t\}$.

Suppose that $(G_i,k,d)$ is a YES-instance of  \textsc{Small Clique in a Regular Graph} for some $i\in\{1,\ldots,t\}$. Then $G_i$ has a clique $K$ of size $k$. Let $\{v_1,\ldots,v_q\}=N_{G_i}(K)$. 
For $j\in\{1,\ldots,q\}$, let $d_j=|N_{G_i}(v_j)\cap K|$. Because $G_i$ is a $d$-regular graph, $d_1+\ldots+d_q=k(d-k+1)=p$. We construct the solution $(U,D,A)$ for $(G\rq{},\delta,d\rq{},k\rq{})$ as follows. We set $U=K$ in the copy of $G_i$, and let $D=\emptyset$. Observe that to satisfy the degree conditions, we have to add $d_j$ edges incident to each $v_j$ in the copy of $G_i$ and add one edge incident to each $u_h$. 
To construct $A$, we consecutively consider the vertices $v_j$ in the copy of $G_i$ for $j=1,\ldots,q$. For each $v_j$, we greedily select $d_j$ vertices $x_1,\ldots,x_{d_j}$ in $\{u_1,\ldots,u_p\}$ that were not selected before and add $v_jx_1,\ldots,v_jx_{d_j}$ to $A$. It is straightforward  to verify that $(U,D,A)$ is a solution and $|U|+|D|+|A|=k+p=k\rq{}$.   

Assume now that $(U,D,A)$ is a solution for  $(G\rq{},\delta,d\rq{},k\rq{})$.

We show that $U\cap(\{u_1,\ldots,u_p\}\cup\{w_0,\ldots,w_k\})=\emptyset$. To obtain a contradiction, assume that $|U\cap(\{u_1,\ldots,u_p\}\cup\{w_0,\ldots,w_k\})|=h>0$.
Let $X=(\{u_1,\ldots,u_p\}\cup\{w_0,\ldots,w_k\})\setminus U$. Because $\{u_1,\ldots,u_p\}\cup\{w_0,\ldots,w_k\}$ has $k(d-k+1)+k+1=k\rq{}+1$ vertices, $X$ has $k\rq{}+1-h>0$ vertices.
Let $G\rq{}\rq{}=G\rq{}-U$.
Observe that for $v\in X$, $\delta(v)-d_{G\rq{}\rq{}}(v)\geq h$. Because the vertices of $X$ are pairwise adjacent, the set $A$ has at least $|X|h=(k\rq{}+1-h)h$ elements. 
But $|A|\leq k\rq{}-|U|\leq k\rq{}-h$. Because $(k\rq{}-h+1)h>k\rq{}-h$, we obtain a contradiction. 

Next, we claim that $|U|=k$ and $D=\emptyset$. Because $U\cap(\{u_1,\ldots,u_p\}\cup\{w_0,\ldots,w_k\})=\emptyset$, $\sum_{j=1}^p(\delta(u_j)-d_{G\rq{}}(u_j))=p$ and the vertices $u_1,\ldots,u_p$ are pairwise adjacent, $A$ contains at least $p$ elements. Moreover, $A$ has at least $p$ edges with one end-vertex in $\{u_1,\ldots,u_p\}$ and another in $V(G_1)\cup\ldots\cup V(G_t)$ for the copies of $G_1,\ldots,G_t$ in  $(G\rq{},\delta,d\rq{},k\rq{})$.
Hence, $|U|+|D|\leq k\rq{}-|A|\leq k\rq{}-p=k$. Suppose that $|U|=s<k$ and $|D|=h$. Let also $D\rq{}=D\cap(E(G_1)\cup\ldots\cup E(G_t))$ and $h\rq{}=|D\rq{}|$. Let $G\rq{}\rq{}=G\rq{}-U-D\rq{}$. 
Because $G_1,\ldots,G_t$ are $d$-regular, $\sum_{v\in V(G\rq{}\rq{})}(\delta(v)-d_{G\rq{}\rq{}}(v))\leq sd+2h\rq{}\leq sd+2h\leq sd+2(k-s)$. Therefore, $A$ contains at most $sd+2(k-s)$ edges with one end-vertex in $V(G_1)\cup\ldots\cup V(G_t)$. Notice that $sd+2(k-s)\leq (k-1)d+2$ because $d>k^2\geq 4$.
But $p-(k-1)d-2=k(d-k+1)-(k-1)d-2=d-k^2+k-2>0$ as $d>k^2$, and we have no $p$ edges with one end-vertex in $\{u_1,\ldots,u_p\}$ and another in $V(G_1)\cup\ldots\cup V(G_t)$; a contradiction. Hence, $|U|=k$ and $D=\emptyset$.

Now we show that $U$ is a clique. Suppose that $U$ has at least two non-adjacent vertices. Let $G\rq{}\rq{}=G\rq{}-U$. Because $G_1,\ldots,G_t$ are $d$-regular, $\sum_{v\in V(G\rq{}\rq{})}(\delta(v)-d_{G\rq{}\rq{}}(v))\geq k(d-k+1)+2=p+2$. 
Recall that $A$ has at least $p$ edges with one end-vertex in $\{u_1,\ldots,u_p\}$ and another in $V(G_1)\cup\ldots\cup V(G_t)$. Because $|U|=k$ and $k\rq{}=p+k$, $A$ consists of 
$p$ edges with one end-vertex in $\{u_1,\ldots,u_p\}$ and another in $V(G_1)\cup\ldots\cup V(G_t)$. But to satisfy the degree restrictions for the vertices of $V(G_1)\cup\ldots\cup V(G_t)$, we
need at least $p+2$ such edges; a contradiction.

We have that $U\subseteq V(G_1)\cup\ldots\cup V(G_t)$ is a clique of size $k$. Because the copies of $G_1,\ldots,G_t$ in  $(G\rq{},\delta,d\rq{},k\rq{})$ are disjoint, $U$ is a clique in some $G_i$.

It remains to apply Theorem~\ref{thm:BJK}. 
Because \textsc{Small Clique in a Regular Graph} is \classNP-complete by Lemma~\ref{lem:clique},  \textsc{Editing to a Graph of Given Degrees}  parameterized by $k+d$ has no polynomial kernel unless  $\classNP\subseteq\classCoNP/\text{\rm poly}$.

To prove the theorem for $S=\{\text{vertex deletion},\text{edge addition}\}$, it is sufficient to observe that for the constructed instance $(G\rq{},\delta,d\rq{},k\rq{})$ of \textsc{Editing to a Graph of Given Degrees}, any solution $(U,D,A)$ has $D=\emptyset$, i.e., edge deletions are not used. Hence, the same arguments prove the claim. 
\end{proof}

\section{Conclusion}
We proved that \textsc{Editing to a Graph of Given Degrees} is \classFPT\ when parameterized by $k+d$ for $\{\text{vertex deletion}, \text{edge addition}\}\subseteq S\subseteq\{\text{vertex deletion},\text{edge deletion},\text{edge addition}\}$, but does not admit a polynomial kernel. Our algorithm runs in time $2^{O(kd^2+k\log k)}\cdot poly(n)$ for $n$-vertex graph. Hence, it is natural to ask whether this running time could be improved. Another open question is whether the same random separation approach could be applied for more general variants of the problem. Recall that  Mathieson and Szeider~\cite{MathiesonS12} proved that the problem is \classFPT\ for the case  when vertices and edges have costs and the degree constraints are relaxed: for each $v\in V(G')$, $d_{G'}(v)$ should be in a given set $\delta(v)\subseteq \{1,\ldots,d\}$. It would be interesting to construct a feasible algorithm for this case. 
Notice that a solution $(U,D,A)$ can have a more complicated structure if $U=\emptyset$. In particular, we cannot claim that $H(D,A)$ can be covered by $(D,A)$-alternating trails.
Some interesting results in this direction were recently obtained by Froese, Nichterlein and Niedermeier~\cite{FroeseNN14}.

\end{document}